\documentclass[conference]{IEEEtran}
\IEEEoverridecommandlockouts
\usepackage{cite}
\usepackage{amsmath,amssymb,amsfonts}
\usepackage{graphicx}
\usepackage{textcomp}
\usepackage{xcolor}
\def\BibTeX{{\rm B\kern-.05em{\sc i\kern-.025em b}\kern-.08em
    T\kern-.1667em\lower.7ex\hbox{E}\kern-.125emX}}
\usepackage[utf8]{inputenc}
\usepackage{amsthm}
\usepackage{algorithm, algorithmic}
\usepackage{hyperref}

\newtheorem{theorem}{Theorem}
\newtheorem{lemma}{Lemma}

\newtheorem{definition}{Definition}
\newtheorem{formula}{Formula}

\begin{document}

\title{Interference in Wireless Networks - A Power Allocation Approach}

\author{\IEEEauthorblockN{1\textsuperscript{st} Tzalik Maimon}
\IEEEauthorblockA{\textit{Ceragon Networks}\\
tzalikm@ceragon.com}
\and
\IEEEauthorblockN{2\textsuperscript{nd} Shirley Alus}
\IEEEauthorblockA{\textit{Ceragon Networks}\\
shirleya@ceragon.com}
\and
\IEEEauthorblockN{3\textsuperscript{rd} Gil Kedar}
\IEEEauthorblockA{\textit{Ceragon Networks}\\
gilk@ceragon.com}
}

\maketitle

\begin{abstract}
{\em Co-Channel Interference (CCI)} is a fundamental problem in wireless communication networks. It is a well-studied problem in the field. As channels use the same frequency, interference in the radio waves occurs which, in turn, reduces the capacity of the interfered channels. There is a need to use the least number of frequencies as communication networks advance to 5G. In this paper, we present a novel technique to manage interference on channels. We use time division for links of the same frequency and, as a result, we show a significant reduction in the number of frequencies used overall in the network. 
\end{abstract}

\section{Introduction}

The next generation of mobile networks sets several challenges, such as, among others, higher sensitivity to interference. As networks become more dense and the band of frequencies used is limited in size, {\em Co-Channel Interference (CCI)} becomes a core issue. Links that use the same frequency may interfere with each other's transmitting signals. This inevitably causes a loss of capacity in the network. Thus, lowering interference in the network or managing it is a well-studied problem in the field.

The approach we take in this paper is that of {\em Power Allocation (PA)}. We change the power level transmitted from each link in constant intervals such that high priority links are less interfered in certain time slots. The power transmitted by interfering links is calculated to be just under an allowed level of interference using Free Path Space Loss and Radiation Patterns of radio antennas. We formulate the problem using graph theory and show that using our method, we can reduce the number of frequencies required to achieve the same capacity.
Previous works tried to allocate the power using machine learning techniques \cite{DEFB21, HZ20, SLGG22, SZMC2021, WLGK18}.  The advantage of our technique is that machine learning requires training while our deterministic approach does not. Another advantage is that changing power levels over time, instead of using a fixed level, gives our algorithm a greater range of possible interference values, each depending on requirements or a goal to achieve (average, minimum, maximum, etc). 
Another known method is using Nash equilibrium \cite{WL22, ZSBJD16, ZSBKJ22, TLCH19}. These algorithms also balance constant output power levels which still have very high interference between links. Also, sometimes equilibrium cannot be reached which cannot be the case in our method. Since our algorithm aims to minimize interference to a certain level, we outperform Nash equilibrium algorithms.

\section{Our Method}

For measuring the capacity in a network, we use the formula of Shannon and Hartley \cite{S98} given as $C = B \cdot \log_2(1+ \frac{P_r}{N})$ where $C$ is the capacity, $B$ is the bandwidth, $P_r$ is the signal power received and $N$ is the noise power. As part of our method, we use the relation between the incoming power in a receiver and the outgoing power in an antenna. This is given by Friis \cite{S13} as $\frac{P_r}{P_t} = D_r D_t (\frac{c}{4 \pi d f})^2$ where $P_r$ is power received, $P_t$ is power transmitted, $D_r$ is the directivity of the receiver, $D_t$ is the directivity of the transmitter, $c$ is the speed of light, $d$ is the distance between the transmitter and the receiver, and $f$ is the frequency of the signal. Therefore, we have the following formula.

\begin{formula}  \label{thm:CapacityMain}
{\bf : The Power-Capacity Function (PC)} $C = B \cdot \log_2(1+ \frac{P_t D_r D_t c^2}{N(4 \pi d f)^2})$.
\end{formula}

\noindent We define the angle at which the signal enters an antenna as the Angle of Reception (AoR). We define the Radiation Pattern Function (RPF) as the function that per angle and distance describes the gain of the signal received from the source. Let $\theta$ be the AoR. Then, the power received is $P'_r = P_r \cdot e^{-\frac{4\theta^2}{\sqrt{2}w^2}}$. This will be important to amplify interferring links.

The algorithm we perform syncs links using their internal clocks. We can thus assume that we are allowed to refer to a "global clock" for all links in the network. The synchronized queue is expressed in time slots (or ticks) each of the same length. At time slot (tick) $j$ all edges which were assigned the label $j$ are considered the highest priority. For this purpose, we consider the communication network as the input multigraph of a coloring algorithm. We define a graph where each vertex $v \in V$ is a node in the network. The edges of $G$ are defined by two types. Black edges, denoted in a subset $E_b$, represent a link between two antennas. A black edge is directed in the direction of the link it represents. A red edge $e$ is directed from a vertex $v$ to a vertex $u$ if a link $e_1 = (v,u)$ interferes with a link $e_2 = (w,u)$. We refer to $e_1$ as the {\em base} of $e$. We define $E(G) = E_b \cup E_r$. See Figure \ref{fig:InputGraph} for an example of an input graph.

\begin{figure}[ht]
\caption{Example of Input Graph}
\centering
\includegraphics[width=0.5\textwidth]{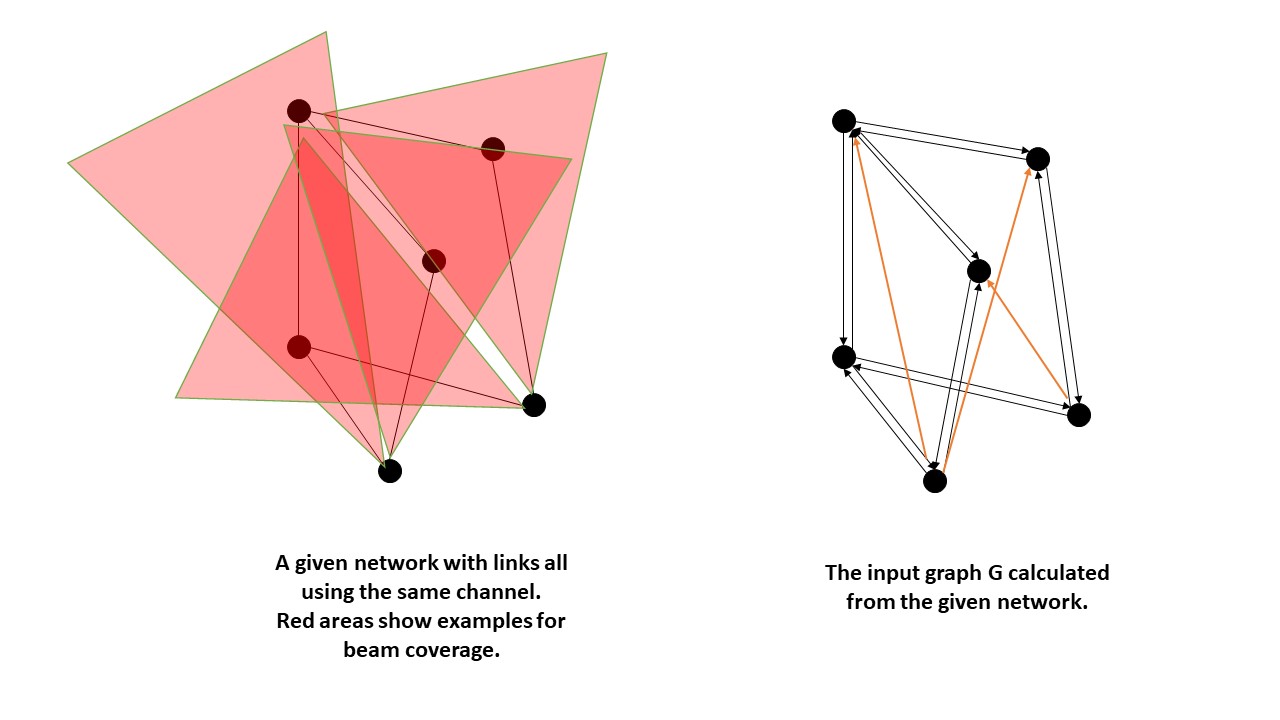}
\label{fig:InputGraph}
\end{figure}

\subsection{Building a Queue}  \label{sec:colorAlgo}

In this section, we color the input graph $G$. We say that a red edge $e'$ is {\bf an effect} of $e$ if $e$ is the link for which the interference $e'$ occurs. We denote all the effects of $e$ as $r(e)$. In the same manner, we define $e$ as the {\bf base} of $e'$. We denote it $b(e')$. Once we have colored a black edge $e$ in some color $c$, then all effects $r(e)$ are assigned the same color. If we color a red edge $e$ with a color $c'$, then its base $b(e')$ is colored using the same color, and likewise for the effect set $r(b(e'))$. To achieve this, we build a new graph $H$ on which we represent the dependencies above. We then reduce the coloring of $G$ to the coloring of $H$. The coloring of $H$ is described at Algorithm \ref{alg:EdgeColoring}. \\

\begin{algorithm}
\caption{DependentEdgeColoring($G$)}
\label{alg:EdgeColoring}
\begin{algorithmic}
\STATE initiate $H$ as an empty graph.
\FOR{$e \in E_b(G)$}
    \STATE Add a vertex $v(e)$ representing $e$ to $H$.
\ENDFOR
\FOR{$e' \in E_r(G)$}
    \STATE $s \gets$ the base of $e'$.
    \STATE $t \gets$ the link interfered by $e'$
    \STATE Add an edge between $v(s)$ and $v(t)$ to $H$.
\ENDFOR
\STATE Vertex-Color $H$.
\FOR{$v \in V(H)$}
    \STATE Let $e \in E_b(G)$ represented by $v$. Then the color of $e$ in $G$ equals the color of $v$ in $H$.
\ENDFOR
\end{algorithmic}
\end{algorithm}

The following lemma shows that the coloring of $M$ is a legal queue for $G$. 

\begin{lemma}   \label{lem:legalQueue}
The legal coloring of $M$ colors the edges in $G$ such that: \begin{itemize}
    \item For each $h_i \in H$ no two edges in $h_i$ are of the same color.
    \item For each $c_i \in C$ all edges in $c_i$ have the same color.
\end{itemize}
\end{lemma}
\begin{proof}
Denote $e_1, e_2$ two edges representing interfering links upon a vertex $v$. Then there is a permutation $h_i \in H$ such that $e_1, e_2 \in h_i$. Then, upon constructing $M$, there is an edge $(e_1, e_2) \in E(M)$. W.L.O.G. assumes $e_1$ is a red edge. Then there is an edge $(b(e_1),e_2)$ in $E(M)$. When we use the coloring of $M$ on $G$, we have $L(e_1) = L(b(e_1)) \neq L(e_2)$. Thus, $e_1$ and $e_2$ have distinct colors in $G$. This can be shown also in the same way if both $e_1, e_2$ are red edges or both are black edges. \\
As for the second requirement, it is clear that every $c_i \in C$ is composed of one black edge that is an actual link and all other edges are the effects of that link. Therefore, the base and all of its effects are colored using the same color from the replacement of red edges with black edges in $H$ while constructing $M$.
\end{proof}

\section{Dynamic Signal Power}  \label{sec:DynamicPower}

We calculate the signal power for each link in each slot. In each slot $j$, each edge $e \in E_b$ is allowed to transmit at signal strength at most $P_j(e)$. During this calculation, we create a queue between incoming edges of each vertex $v \in G$ depending on the coloring $L$ we computed in section \ref{sec:colorAlgo}. Let $e = (v,u)$ be a black edge. Let $S$ be the set of edges with which $e$ interferes. Then the power upon $e$ in time slot $j$ is defined as

\begin{equation}  \label{eq:power}
    P(e)_j= 
        \begin{cases}
            P_{max}        & \text{if } L(e_i) = j(Q) \\
            P_\ell(j)     & \text{otherwise}
        \end{cases}
\end{equation}

\noindent where $P_{max}$ is the full power and $P_\ell(j)$ is a reduced power. See Figure \ref{fig:PowersQueue} for an example of powers assigned depending on the queue.

\subsection{Calculating the Reduced Power} 

Let $u$ be a vertex with a receiver on which we have interference from an edge $e = (v,w)$. If there is no such $u$, we define $P_\ell(j) = P_{max}$. Denote $d$ the distance from $v$ to $u$. For the purpose of considering AoR, we denote the RPF as $z$. Denote $\theta$ the AoR of $e$ upon $u$. If $e$ interferes with several links, we denote $u$ the terminal such that $d^2 z^{-1}(\theta)$ is the smallest. We will aim for the maximum of $P_{min}$ power to reach $u$. Therefore,

\begin{equation} \label{eq:loweredPower}
\begin{split}
    P_{min} &= \frac{P_\ell(j) \cdot z(\theta)}{4\pi d^2} \\ 
    P_\ell(j) &= (4 \pi d^2)z^{-1}(\theta) \cdot P_{min}
\end{split}
\end{equation}

But this is true if $e$ is the only edge for which $u$ is an exposed terminal in time slot $j$ for some link $l$. We call $e$ a {\bf interfering edge} on $l$. We define a set $B^j_l$ of all interfering edges on $l$ in slot $j$. W.L.O.G. assumes that the capacity is shared equally among all edges $\{b(e)\ |\ e \in B^j_l\}$. Thus, we have

\begin{equation} \label{eq:reducedPower}
    P_\ell(j) = \frac{(4 \pi d^2)z^{-1}(\theta) \cdot P_{min}}{|B^j_l|}
\end{equation}

\begin{figure}[!htb]
    \centering
    \caption{Example of Power Assignments}
    \includegraphics[width=0.5\textwidth]{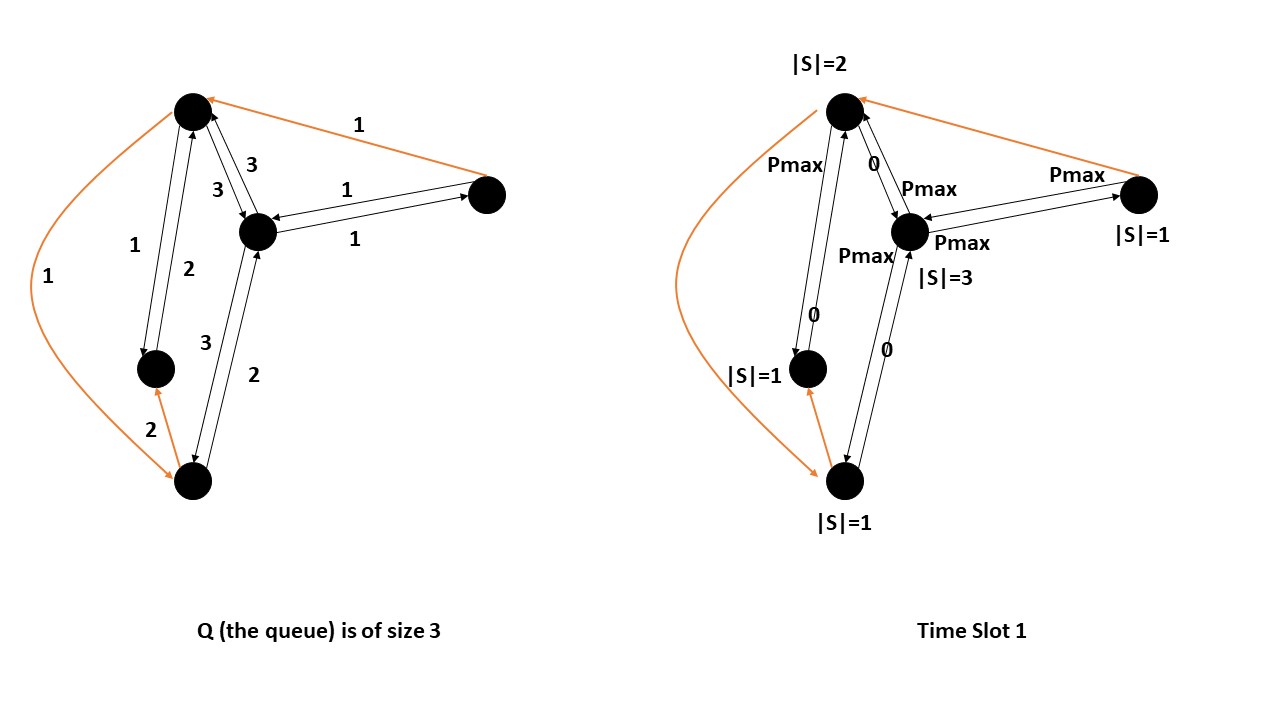}
    \label{fig:PowersQueue}
\end{figure} 
\begin{figure}[!htb]
    \centering
    \includegraphics[width=0.5\textwidth]{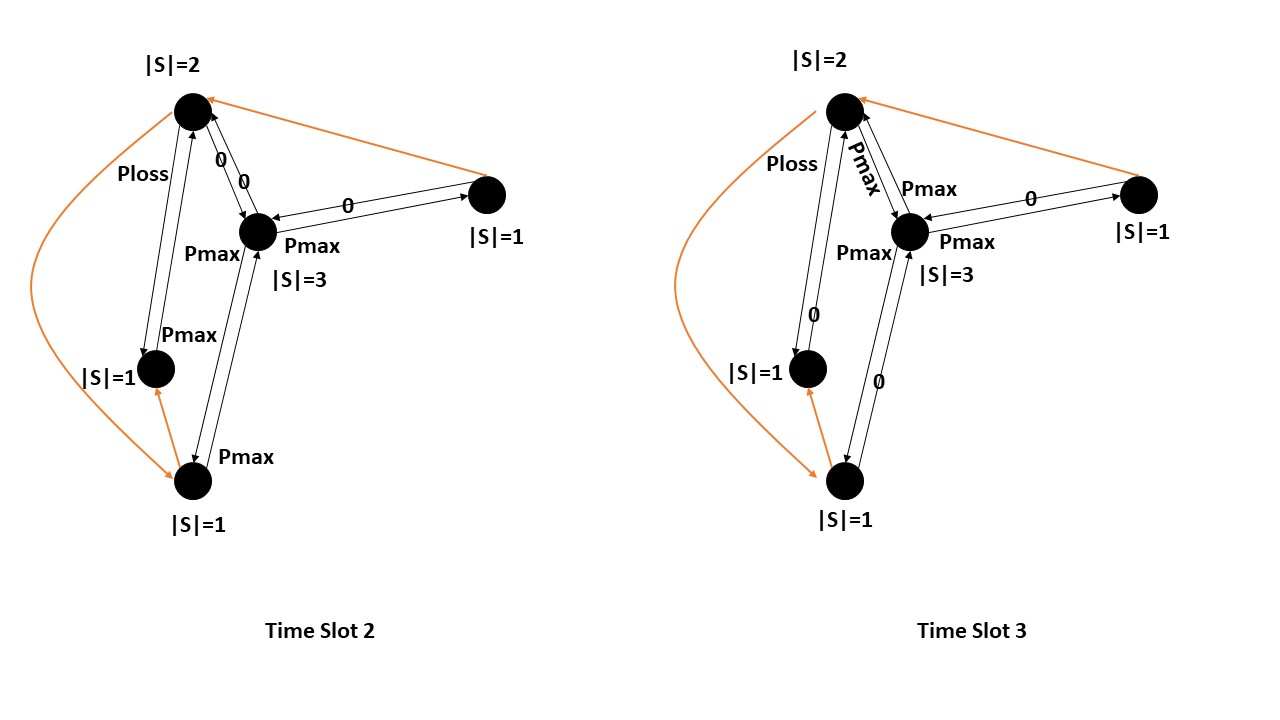}
\end{figure}

\section{Allowed Interference Interpolation}

The solution proposed in section \ref{sec:DynamicPower} assumes some allowed interference $P_{min}$. This parameter was considered as the interference threshold which is undetectable in the signal due to technology. But we can allow more interference which, on the one hand, causes more interruption in priority edges, but on the other hand, allows more power on blocked edges.  

Let $C(G, Q)$ be the function that measures the capacity of the entire network using the queue $Q$. We define the function $f(x)$ where $x$ is a parameter denoting a ratio of allowed interference. That is, for each link $e$ we allow the interference in time slot $j$ to be $I(E) \leq \frac{P_r(j)}{x}$. Therefore, $f$ is defined as running the algorithm for dynamic power where $P_{min} \leq I(E)$ for each edge. The function then returns $C(G, Q)$ as a result. 

Calculating the optimal value for $x$, denoted $x_{m}$, starts with setting all powers of links in $G$ to their maximum. We then check the smallest value $x_1$ of SIR across all links. We define the lowest possible value of allowed interference as a parameter $x_0 > 0$, which can be as small as we wish. Let $k > 0$ be an input parameter. We split the close segment $[x_0, x_1]$ into $k-1$ sub-segments such that we have $k$ points $\{x_i\}$ at constant intervals. We calculate $f(x)$ on these $k$ points. Denote $\{y_i\}$ the set of results. We use interpolation methods to approximate a function on the set $\{(x_i, y_i)\}$. We find a maximum point $\hat x_i$. Since the interpolation is an approximation of the actual function $f(x)$ over the segment, we calculate the actual value of $f(\hat x_i)$ and compare them to the values in $\{y_i\}$. Let $y_m$ be the maximum value in the set $\{y_i\} \cup \{f(x)|x \in \{\hat x_i\})\}$. We define $x_m$ to be the $x$ value corresponding to $y_m$. We now use recursion on the above scheme to better approximate the optimal value for the allowed interference. Once we choose a value $x_m$, we continue by defining a new close segment around $x_m$ and repeat the above scheme on it. The new sub-segment is defined as $[x_m - \frac{x_1 - x_0}{2(k-1)}, x_m + \frac{x_1 - x_0}{2(k-1)}]$. We can repeat this for as many recursion levels as we wish until the value of $x_m$ does not change between recursions. We then execute the dynamic power algorithm with $x_m$ as a ratio for allowed interference where $P_{min}$ is calculated accordingly for each edge $e$.

\section{Frequency Allocation Planning}

We have, so far, assumed that the given input graph $G$ is a result of the shuttering of a network into links that share the same frequency and therefore require solving interference. The frequencies are assigned to the links during the planning stage. In this section, we further extend the ideas we presented and propose algorithms for frequency assignment. The motivation here is to allocate extra resources (more frequencies) accordingly for optimization of the profit from additional channels.  

We start with a single frequency for all links in the network. Then, we add a single frequency and evaluate the profit from such an addition. If the profit is more than a predefined parameter, we assign the frequency to the links which will produce the said profit, and try to add another frequency. We repeat this until the profit we gain is smaller than a given threshold. Note that the profit is a function that we can define as desired. In this paper, we refer to the profit as a function of the total capacity transmitted in the network per second. We use the solution in section \ref{sec:DynamicPower} to assign the new frequency to edges where the benefit of such a change would be most profitable. We thus create two sub-graphs of the original network and can recursively repeat the process for using more frequencies. \\

To evaluate a benefit, we define the Power-Gain (PG) function of each edge $e = (v,u)$. The function measures the amount of additional power that would translate to capacity in the network given that $e$ changes its frequency to the new one. Let $G_1$ be the graph of the current frequency and $G_2$ be the graph of the new frequency. We define the value $P_{Q}(e)$, which is the average power a black edge $e$ transmits during a queue $Q$. For the case of a red edge $e$ we define $P_{Q}(e) = - \frac{\sum_{j \in Q} Pr_j(e)}{Q}$, which is the actual gain of the interference that will be added to the power translated to capacity in case the red edge is removed from the graph. The gain and loss of power of $e$ can then be measured as follows. Let $Q_1$ be the queue in $G_1$ and $Q_2$ be the queue in $G_2$ after $e$ moves to $G_2$. Then we can give a formal definition of $PG(e)$.

\begin{definition}  \label{def:PowerGainFucntion}
Let $e = (v,u)$ be an edge. Let $\overrightarrow E^r_{G_1}$ denote the red edges contesting $e$ on $u$ and $b(\overrightarrow E^r_{G_1})$ be the base edges of these red edges. We define $\mathcal{U}_{G_1} = \overrightarrow E^r_{G_1} \cup r(e) \cup b(\overrightarrow E^r_{G_1})$ the surrounding edges of $e$. (See Figure \ref{fig:Vicinity} for an example of the vicinity $\mathcal{U}$ of an edge.) The power gain function of $e$ in $G_1$ is defined by 

\begin{equation} \label{eq:pg1}
    PG_{G_1}(e) = - P_{Q}(e) + \sum_{e_i \in \mathcal{U}_{G_1}} (P_{Q_1}(e_i) - P_{Q}(e_i))    
\end{equation}

Here, $Q$ is the queue in $G_1$ prior to removing $e$.
In the same manner, we define $\mathcal{U}_{G_2} = \overrightarrow E^r_{G_2} \cup r_2(e) \cup b(\overrightarrow E^r_{G_2})$ the surrounding edges of $e$ in $G_2$. Here $r_2(e)$ are the effects of $e$ in $G_2$. The power gain function of $e$ in $G_2$ is defined by 

\begin{equation} \label{eq:pg2}
    PG_{G_2}(e) = P_{Q_2}(e) + \sum_{e_i \in \mathcal{U}_{G_2}} (P_{Q_2}(e_i) - P_{Q'}(e_i))    
\end{equation}

Here $Q'$ is the queue in $G_2$ before adding $e$.
Now we can define the PG function over the edge set of $G$,

\begin{equation} \label{eq:totalPG}
    PG(e) = PG_{G_1}(e) + PG_{G_2}(e)
\end{equation}
\end{definition}

\begin{figure}[ht]
\caption{Example of the set $\mathcal{U}$ for the link in green.}
\centering
\includegraphics[width=0.5\textwidth]{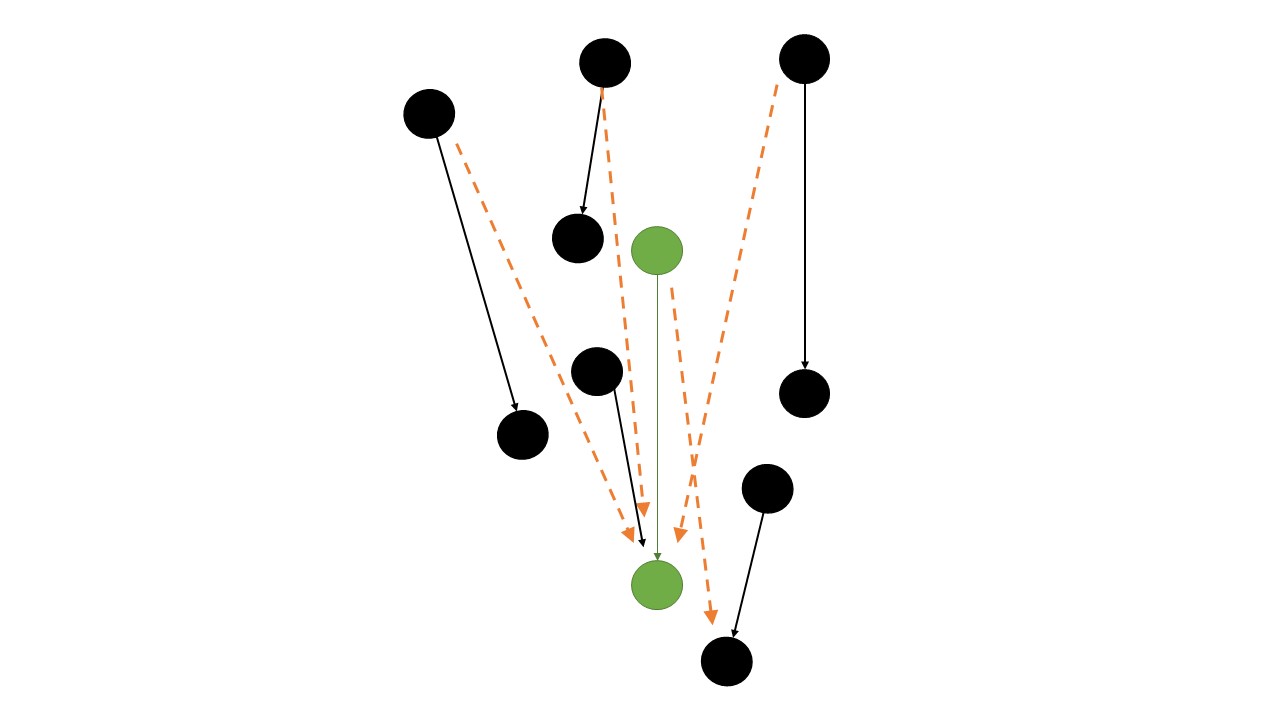}
\label{fig:Vicinity}
\end{figure}

\subsection{Greedy is Optimal}

We define the optimality of a subgraph using the PG values of its edges. The optimality of an assignment to the new frequency is measured by the amount of power transmitted across each edge during a queue minus the amount of interference during the transmission. 
By the definition of the PG function, the $PG$ values of edges in $G_1$ contain within them the loss we will have in $G_2$ as well. Thus, we can define the optimality of power-interference of both graphs as minimizing the average of $PG$ values on edges in $G_1$. That is, as long as there is an edge with $PG > 0$ in $G_1$, we still have something to gain from moving it to $G_2$. Our algorithm continues to move edges to the new frequency as long as there are such edges with $PG > 0$. Therefore, the average of $PG$ values in $G_1$ goes to 0 as the algorithm progresses and our algorithm is optimal with regard to a power interference ratio.

\subsection{Runtime Analysis }

We analyze the running time for computing the PG value of an edge by definition \ref{def:PowerGainFucntion}. To calculate the values of $P_{Q_1}$ and $P_{Q_2}$, we require one edge addition and one edge deletion. This is achieved within at most $5$ operations. For the sum $PG_{G_1}$ we require at most $\delta$ operations. For the sum $PG_{G_2}$ we require at most $\delta+2$ operations. One more operation is required for summing the value $PG$. Overall, the calculation of the PG function requires $O(\delta)$ time. \\

We next present an analysis of the overall runtime of our frequency assignment algorithms. We detail the steps and the required running time. We denote $m = |E_b(G)|$.

\begin{enumerate}
    \item We require $O(\delta m)$ time for building the queue.
    
    \item Calculating the power tables on $G$ also requires $O(\delta m)$. 

    \item We calculate the PG values of black edges in $G$. This requires $O(\delta m)$.
    
    \item We sort the edges by PG values in $O(m \log (m))$ time and select the maximum.
    
    \item We update the PG value of at most $\delta$ edges (only black edges in $\mathcal{U}_{G_1}$). This requires $O(\delta^2)$ time.
    
    \item We need to sort the updated edges within the list of all edges. This requires $O(\delta \log (m))$ time.
    
    \item We repeat steps 4 and 5 until we have no more edges with a positive PG value. A thorough analysis can show a small running time, but executing at most $O(m)$ iterations is enough for an efficient analysis.
    
    \item We reduce the labeling to be greedy by iterating all black edges in $G_1 \cup G_2$ and check for the smallest available label for each edge. This requires $O(\delta m)$ time.  
    
    \item We recalculate the power tables for $G_1$ and $G_2$, which also takes $O(\delta m)$ time.
\end{enumerate}

From the above, we can state the following theorem to conclude our algorithm.

\begin{theorem}
There is an $O(\delta m \cdot (\log (m) + \delta))$ deterministic algorithm for assigning a new frequency to a given graph.
\end{theorem}

As technology advances, networks require more links. Therefore, there are more black edges in the input graph and interference is better managed at the hardware level, meaning fewer red edges in $G$. Thus, it is more likely that $\log(m) > \delta$, and so our algorithm behaves as $\tilde O(\delta m \log(m))$. This is only greater by a factor of $\delta$ than the lower bound of any weighted assignment algorithm.

\section{Experimental Results}

We simulated a communication network where all links use the same frequency. We generated random graphs and for each generated graph we compared the network performance with and without our method. We ran the experiments on a Python simulator using the library Numpy. We performed three experiments, each on 10 randomly generated graphs, and we present here the average on each measurement. Each experiment denotes $V$ as the number of vertices, $E$ as the number of links, and $D$ as the maximum allowed degree for each vertex (which affects the density of the network). Table \ref{tab:ExpResults} shows the percentage of change achieved after using our algorithm.

\begin{table}[h!]
\centering
\begin{tabular}{||p{0.3cm}|p{0.3cm}|p{0.3cm}|p{1cm}|p{1cm}|p{1cm}|p{1cm}||} 
 \hline
 V & E & D & Total Capacity & Best Improvement & Power Used & Capacity Loss Due to Interference \\  
 \hline
 20 & 30 & 10 & 234\% & 179\% & 33\% & 7\%  \\
 \hline
 30 & 100 & 30 & 309\% & 179\% & 13\% & 3\%  \\
 \hline
 100 & 300 & 80 & 326\% & 240\% & 5\% & 1\%  \\
 \hline
\end{tabular}
 \vspace{3pt}
\caption{Experiment Results for Dynamic Power}
\label{tab:ExpResults}
\end{table}

Our algorithm not only achieves our main goal of increasing capacity on a single channel but also prevents the waste of energy. We tested the algorithm on two planned networks, one with 8 frequencies used and the second with 4 frequencies used. In each case, we showed that the capacity achieved in these networks using the planned number of frequencies can be achieved with fewer frequencies. Moreover, we also showed that our scheme significantly increases the overall capacity of the network when the originally planned number of frequencies is used. These results are presented in Figures \ref{img:napal} and \ref{img:india}.

\begin{figure}[!htb]    
    \centering
    \includegraphics[width=0.5\textwidth]{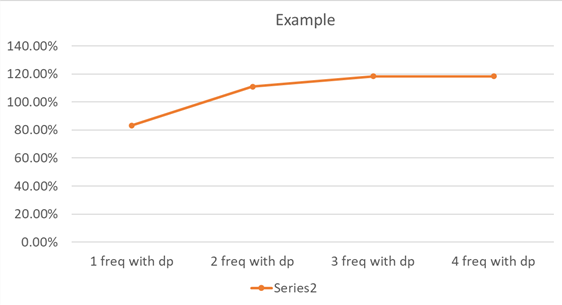}
    \caption{A network of 24 links and 4 frequencies used}
    \label{img:napal}
\end{figure}

\begin{figure}[!htb]  
    \centering
    \includegraphics[width=0.5\textwidth]{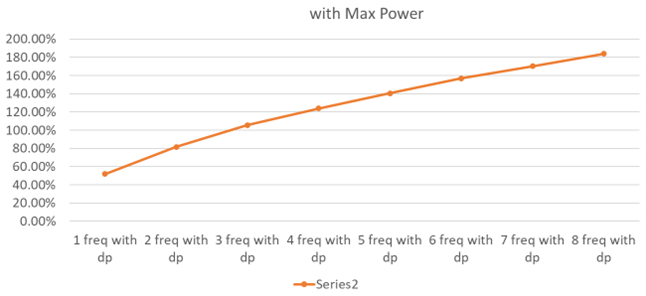}    
    \caption{A network of 366 links and 8 frequencies used}
    \label{img:india}
\end{figure}

\section{Conclusion}

The problem of interference in wireless communication networks is a long-standing problem. Solutions in the past mainly focused on probabilistic machine learning techniques. We devised a novel technique for managing interference. Using this technique we also devised a frequency assignment algorithm which optimized the benefit of adding more frequencies to the network. In experiments, we see that our technique lowers the number of frequencies required for a demand of a network as well as drastically decreasing the wasted energy. The approach of time allocation for transmitted power opens a new direction for solving this problem. As hardware evolved with time, time allocation can be used for frequency changes in each time slot, angle changes in each time slot as well as the existing power change in each time slot. This will give us three dimensions to lower the required resources even further.

\bibliographystyle{plain}
\bibliography{refs}

\end{document}